\documentclass{article}

\usepackage{arxiv}

\usepackage[utf8]{inputenc} 
\usepackage[T1]{fontenc}    
\usepackage{hyperref}       
\usepackage{url}            
\usepackage{booktabs}       
\usepackage{amsfonts}       
\usepackage{nicefrac}       
\usepackage{microtype}      
\usepackage{lipsum}

\usepackage{cite}
\usepackage{amsmath,amssymb,amsfonts}
\usepackage{graphicx}
\usepackage{algorithm}
\usepackage[noend]{algpseudocode}
\usepackage{hyperref}
\usepackage{textcomp}
\usepackage{threeparttablex} 

\usepackage{mathtools}
\usepackage{caption}
\usepackage{float}
\usepackage{stmaryrd}
\usepackage{epstopdf}
\usepackage{multirow}
\usepackage{pifont}
\usepackage{caption}
\usepackage{subcaption}

\newcommand{\R}{{\mathbb{R}}}

\newcommand{\N}{{\mathbb{N}}}

\newcommand{\X}{{\mathbf{X}}}

\newcommand{\con}{{\mathcal{o}}}

\newcommand{\F}{{\lozenge}}
\newcommand{\G}{{\Box}}
\newcommand{\W}{{\mathcal{W}}}
\newcommand{\true}{{\mathsf{true}}}
\newcommand\mydots{\hbox to 1em{.\hss.\hss.}}

\newtheorem{theorem}{Theorem}[section]

\newtheorem{assumption}{Assumption}

\newtheorem{definition}[theorem]{Definition}

\newtheorem{remark}[theorem]{Remark}

\newtheorem{proof}[theorem]{Proof}

\title{Approximation-free Control for Signal Temporal Logic Specifications using Spatiotemporal Tubes
\thanks{ This work was supported in part by the SERB Start-Up Research Grant; in part by the ARTPARK. The work of Ratnangshu Das was supported by the Prime Minister’s Research Fellowship from the Ministry of Education, Government of India.}
}

\author{
 Ratnangshu Das \\
  Robert Bosch Centre for Cyber-Physical Systems\\
  IISc, Bengaluru, India\\
  \texttt{ratnangshud@iisc.ac.in} \\
   \And
 Subhodeep Choudhury \\
  Department of Electrical Engineering,\\
  BITS Pilani, Goa,\\
  \texttt{f20212913@goa.bits-pilani.ac.in} \\
  \And
 Pushpak Jagtap \\
  Robert Bosch Centre for Cyber-Physical Systems\\
  IISc, Bengaluru, India\\
  \texttt{pushpak@iisc.ac.in} \\
}

\begin{document}
\maketitle

\begin{abstract}
This paper presents a spatiotemporal tube (STT)-based control framework for satisfying Signal Temporal Logic (STL) specifications in unknown control-affine systems. We formulate STL constraints as a robust optimization problem (ROP) and recast it as a scenario optimization program (SOP) to construct STTs with formal correctness guarantees. We also propose a closed-form control law that operates independently of the system dynamics, and ensures the system trajectory evolves within the STTs, thereby satisfying the STL specifications. The proposed approach is validated through case studies and comparisons with state-of-the-art methods, demonstrating superior computational efficiency, trajectory quality, and applicability to complex STL tasks.
\end{abstract}

\section{Introduction}
Control objectives in robotics have traditionally focused on stabilization and trajectory tracking. However, modern applications demand more complex task specifications, where systems must operate under time-dependent constraints and strict deadlines. Signal Temporal Logic (STL) \cite{STL_OdedMaler} provides a formal framework to specify such tasks, making it an increasingly popular tool in areas such as neural networks \cite{STL_NN, STL_RNN}, deep learning \cite{STLnet}, reinforcement learning \cite{STL_RL}, learning from demonstration \cite{STL_lfd}, and planning and control \cite{STL_planning, STL_multi_chuchu}. Although STL offers significant advantages, it presents algorithmic and computational challenges, making control under STL specifications particularly difficult.

STL-constrained control is often posed as an optimization problem. Mixed-integer programming (MIP) \cite{STL_multi_chuchu, STL_reactive_MILP} and gradient-based approaches \cite{STL_grad} handle simple dynamics and structured predicates but scale poorly. 
Model Predictive Control (MPC) \cite{STL_MPC, STL_MPC2, STL_MPC3} incorporates STL into a receding horizon framework, but becomes intractable with increased complexity. Furthermore, these optimization-based methods require accurate models, limiting their application in unknown real-world settings.
Control barrier functions (CBFs) \cite{CBF} enforce STL constraints by ensuring forward invariance of safe sets \cite{CBF_STL_Belta, CBF_STL_Dimos}, offering better scalability, but still depend on known dynamics and involve optimization. Prescribed performance control (PPC) \cite{PPC1} provides a more computationally efficient alternative for STL tasks \cite{STL_PPC, STL_funnel}, particularly for unknown systems, but lacks optimality. Additionally, both CBF-based and PPC-based methods support only restricted STL fragments.

This work proposes a spatiotemporal tube (STT)-based control framework \cite{STT} that addresses these challenges by offering a computationally efficient and formally correct approach to synthesizing controllers that satisfy STL specifications. STTs have been effective in handling temporal reach-avoid-stay (T-RAS) specifications \cite{DDSTT_arxiv}, multi-agent planning in unknown systems \cite{STT_multi} {and can also serve as a foundation for synthesizing CBFs with smooth box-shaped level sets \cite{STT_CBF}}. 
Here, we extend the STT framework to general STL specifications through three key steps: (i) first, frame the STL constraints as robust optimization problems (ROP); (ii) next, reformulate the ROP into a scenario optimization program (SOP) by sampling points in time and state space, enabling efficient construction of STTs with formal correctness guarantees; (iii) finally, derive an approximation-free, closed-form control law that ensures the system state remains within the STTs, guaranteeing STL satisfaction. This framework eliminates the need for system discretization or real-time optimization, significantly improving computational efficiency. Unlike existing methods, it can handle general STL formulas for unknown systems, making it broadly applicable. We validate our approach through case studies and compare it with MILP, MPC, CBF, and funnel-based methods. Empirical results show that our method achieves better control synthesis time, while handling more complex STL specifications for unknown dynamics.

\section{Preliminaries and Problem Formulation} \label{sec:prob}
\textit{Notations}: For $a,b\in\N$ with $a\leq b$, we denote the closed interval in $\N$ as $[a;b] := \{a, a+1, \ldots, b\}$. We use $x\circ y$ to represent the element-wise multiplication where $x,y\in \R^n$. All other notation in this paper follows standard mathematical conventions.

\subsection{System Definition}
Consider the following control-affine nonlinear system:
\begin{align}
    \mathcal{S}: \dot{x} = f(x) + g(x)u + w, \label{eqn:sysdyn}
\end{align}
where {$x(t) = [x_1(t), \ldots, x_n(t)]^{\top} \in \X \subset \mathbb{R}^n$ and $u(t) \in \mathbb{R}^n$ are the state and control input vectors of equal dimension, thereby assuming a fully actuated system.} $w(t) \in \mathbb{W} \subset \R^n$ is the unknown bounded disturbance. The state space of the system is defined by the compact set $\X$.
\begin{assumption}\label{assum:lip}
    The functions $f: \X \rightarrow \R^n$ and $g: \X \rightarrow \R^{n \times n}$ in \eqref{eqn:sysdyn} are unknown and locally Lipschitz continuous.         
\end{assumption}
\begin{assumption}\label{assum:pd}{(\cite{PPC1, PPC0})}
    The symmetric component of $g(x): g_s(x) = \frac{g(x)+g(x)^{\top}}{2}$ is uniformly sign-definite with known sign for all $x \in \X$. Without loss of generality, assume $g_s(x)$ is uniformly positive definite, i.e., there exits $\underline{g}\in\mathbb R^+$:
    $0 < \underline{g} \leq \lambda_{\min} (g_s(x)), \forall x \in \X,$ where $\lambda_{\min}(\cdot)$ is the smallest eigenvalue of the matrix.
\end{assumption}
\begin{remark}
    We can extend from control-affine to nonlinear dynamics by treating the nonlinear terms as disturbances. 
\end{remark}

\subsection{Signal Temporal Logic (STL)}
Signal Temporal Logic (STL) \cite{STL_OdedMaler} is a formal language used to specify the spatial, temporal, and logical properties of continuous-time signals. The set of STL formulae can be recursively expressed using predicates $\mathsf{p}$. Consider the predicate function $h:\R^n \rightarrow \R$, then 
$\mathsf{p} := {\true} \text{, if $h(s) \geq 0$, and} \ \mathsf{false} \text{ if $h(s) < 0$}.$

An STL formula $\phi$ is recursively defined using predicates, Boolean logic, and temporal operators:
\begin{align}\label{eqn:stl}
    \phi := {\true} \ | \ \mathsf{p} \ | \ \neg \phi \ | \ \phi_1 \land \phi_2 \ | \ \phi_1 \lor \phi_2 \ | \ 
    \G_{[a, b]}\phi \ | \ \F_{[a, b]}\phi \ | \ \phi_1\mathcal{U}_{[a, b]}\phi_2,
\end{align}
where $\mathsf{p}$ is a predicate, and $\phi_1, \phi_2$ are STL formulae. Boolean operators for negation, disjunction, and conjunction are denoted by $\neg, \lor$, and $\land$, respectively. Symbols $\F$, $\G$ and $\mathcal{U}$ represent the temporal operators - eventually, always, and until, respectively. $[a, b]$ is the time interval with $a, b \in \R_{\geq 0}$ such that $a \leq b$.
The satisfaction relation $(x,t)\models\phi$ denotes if a signal $x: \R_{\geq 0} \rightarrow \R^n$, possibly a solution of \eqref{eqn:sysdyn}, satisfies an STL formula $\phi$ at time $t$. The STL semantics \cite{STL_OdedMaler} for a signal $x$ is recursively given by:
\begin{align*}
    &(x,t)\models\mathsf{p} &&\Leftrightarrow \  h(x(t))\ge0, \\
    &(x,t)\models\neg\phi &&\Leftrightarrow \  \neg((x,t)\models\phi), \\
    & (x,t)\models\phi_1\land\phi_2 &&\Leftrightarrow \  (x,t)\models\phi_1\land (x,t) \models\phi_2, \\
    &(x,t)\models{\F}_{[a, b]}\phi&&\Leftrightarrow \  \exists t' \in [t+a,t+b], (x,t')\models\phi, \\
    &(x,t)\models{\G}_{[a, b]}\phi&&\Leftrightarrow \  \forall t' \in [t+a,t+b], (x,t')\models\phi \\
    &(x,t)\models \phi_1{\mathcal{U}}_{[a, b]}\phi_2 &&\Leftrightarrow \  \exists t_1 \in [t+a,t+b], (x,t_1)\models\phi \ \land \forall t_2 \in [t+a,t_1], (x,t_2) \models \phi_2
\end{align*}

A signal $x$ satisfies an STL formula $\phi$, denoted $x \models \phi$, if and only if $(x,0) \models \phi$. The STL robustness metric $\rho^{\phi}(x,t)$ \cite{STL_robust} quantifies the degree of satisfaction: $\rho^{\phi}(x,t) \geq 0$ implies $(x,t) \models \phi$, and its magnitude reflects the strength of satisfaction or violation. The robustness semantics are defined recursively as:
\begin{subequations}\label{eqn:stl_rho}
\begin{align}
    &\rho^{\mathsf{p}}(x,t)= h (x(t)), \label{eqn:stta}\\
    &\rho^{\neg\phi} (x,t)= -\rho^{\phi} (x(t)), \\
    &\rho^{\phi_1\land\phi2} (x,t)= \min (\rho^{\phi_1} (x,t), \rho^{\phi_2} (x,t)),\\
    &\rho^{{\F}_{[a, b]}\phi} (x,t)= \max_{t_1 \in [t+a,t+b]} \rho^{\phi} (x,t_1),\\
    &\rho^{{\G}_{[a, b]}\phi} (x,t)= \min_{t_1 \in [t+a,t+b]} \rho^{\phi} (x,t_1),\\
    &\rho^{\phi_1{\mathcal{U}}_{[a, b]}\phi_2} (x,t) = \max _{t_1 \in [t+a, t+b]} \min ( \rho^{\phi_1}(x,t_1), \min_{t_2 \in [t+a,t_1]} \rho^{\phi_2} (x,t_2)).
    \label{eqn:stte}
\end{align}
\end{subequations}
For simplicity, we denote the robustness at $t = 0$ by $\rho^{\psi}(x)$ to represent. {We consider a finite time horizon $[0,t_f]$ over which the specification $\phi$ is to be realized, i.e., the signal $x:[0,t_f] \rightarrow \R^n$ is such that $x \models \phi$. 
}

\subsection{Spatiotemporal Tubes}
To satisfy the given STL specification, we leverage Spatiotemporal tubes (STTs), which are defined next.
\begin{definition}[STTs for STL Task]\label{def:stt}
For an STL task $\phi$ defined in Equation \eqref{eqn:stl} over the time interval $[0, t_f]$, the time-varying intervals $[\gamma_{i,L}(t), \gamma_{i,U}(t)]$ are called STTs for STL, if $\gamma_{i,L}:\R_0^+\rightarrow\R$ and $\gamma_{i,U}:\R_0^+\rightarrow\R$ are continuously differentiable functions satisfying $\gamma_{i,L}(t) < \gamma_{i,U}(t)$ for all $i \in [1;n]$ and for all $t\in[0, t_f]$, and the following holds:
\begin{align}\label{eqn:stt_stl}
    \rho^\phi (x) > 0, \forall x: [0,t_f] \rightarrow \R^n, \ \text{such that, } x(\tau) \in \prod_{i \in [1;n]} [\gamma_{i,L}(\tau), \gamma_{i,U}(\tau)], \forall \tau \in [0, t_f].
\end{align}
\end{definition}

\begin{remark}
    If one designs a control law that constrains the state trajectory within the STTs, i.e.,
    \begin{align} \label{eqn:stt_constrain}
    &\gamma_{i,L}(\tau) < x_i(\tau) < \gamma_{i,U}(\tau), \forall (i, \tau) \in [1;n] \times [0, t_f],
\end{align}
then one can ensure the satisfaction of STL specifications.
\end{remark}

\subsection{Problem Statement}
Given an STL specification $\phi$ over $t \in [0, t_f]$, as in Equation \eqref{eqn:stl}, first synthesize STTs. Then, for an unknown control-affine system \eqref{eqn:sysdyn}, design a closed-form, approximation-free control law $u:X\times\R_0^+\rightarrow\R^n$, that constrains the system state $x(t)$ within the STTs for all $t \in [0, t_f]$, thereby satisfying Equation \eqref{eqn:stt_constrain} and ensuring $x \models \phi$. 

\section{Construction of Spatiotemporal Tubes}\label{sec:sampling_based}
In this section, the main goal is to construct STTs that satisfy the given STL specifications. 
We first fix the structure of the curves that form the STTs for the $i$-th dimension as, 
$$\gamma_{i,\con}(c_{i,\con},t) = \sum_{k=1}^{z_{i,\con}} c_{i,\con}^k p_{i,\con}^k(t), \ \con \in \{L,U\}, \ i\in [1;n],$$
where $L$ and $U$ denote the lower and upper constraints, respectively. $p_{i,\con}(t)$ are user-defined nonlinear continuously differentiable basis functions and $c_{i,\con} = [c_{i,\con}^1; c_{i,\con}^2;...; c_{i,\con}^{z_{i,\con}}] \in \mathbb{R}^{z_{i,\con}}$ denote unknown coefficients.

To satisfy the conditions in Definition \ref{def:stt}, we formulate the following Robust Optimization Program (ROP) for all dimensions $i \in [1;n]$ and for all time $\tau \in [0, t_f]$:
\begin{subequations} \label{eq:ROP}
\begin{align}
& \min_{[d_1, d_2,...,d_n,\eta]} \quad \eta, \quad  \textrm{s.t.}  \notag \\
& \forall (i, \tau) \in [1;n] \times [0, t_f]: \notag \\
& \gamma_{i,L}(c_{i,L},\tau) - \gamma_{i,U}(c_{i,U},\tau) + \gamma_{i,d} \leq \eta, \label{subeqn:ROP1} \\
& \dot{\gamma}_{i,L}(c_{i,L},\tau) - \mathcal{L}_s \leq \eta, \ 
\dot{\gamma}_{i,U}(c_{i,U},\tau) - \mathcal{L}_s \leq \eta; \label{subeqn:ROP2}\\
& \forall x :[0,t_f] \rightarrow \R^n, \text{s.t., }\notag\\
&x(\tau) = [\lambda_1 \gamma_{1,U}(c_{1,U},\tau) + (1 - \lambda_1)\gamma_{1,L}(c_{1,L},\tau), \ldots,  \lambda_n \gamma_{n,U}(c_{n,U},\tau) + (1 - \lambda_n)\gamma_{n,L}(c_{n,L},\tau)]^\top, \notag \\
&\hspace{9cm} \forall (\lambda_{1}, \ldots, \lambda_{n}, \tau) \in [0,1]^n \times [0, t_f], \notag \\
& -\rho^{\phi} (x) \leq \eta \label{subeqn:ROP3}.\\
& d_i = [c_{i,L}, c_{i,U},\eta] \notag.
\end{align}
\end{subequations}
{Here, $\gamma_{i,d} \in \R^+$ and $\mathcal{L}_s \in \R^+$ are user-defined parameters that set the minimum separation and the maximum slope of STTs, respectively.} To facilitate the formulation of the ROP, we express Equation \eqref{eqn:stt_stl} in a parametric form using $\lambda_i \in [0,1]$ for all $i \in [1;n]$ in \eqref{subeqn:ROP3}.
It is evident that if the optimal solution to the ROP satisfies $\eta^* \leq 0$, then the conditions outlined in Definition \ref{def:stt} are guaranteed to hold.
 
{\begin{remark}\label{rem:shortSTT}
    The constraint \eqref{subeqn:ROP2} on the tube derivatives ensures smooth tube evolution, leading to more efficient trajectories with lower control effort for STL tasks. However, it adds complexity and might prevent obtaining solutions even if valid STTs with higher slopes exist. It can be omitted for more flexibility or faster computation.
\end{remark}
}

Solving the proposed ROP in \eqref{eq:ROP} involves infinitely many constraints in continuous space. To address this, we propose a sampling-based approach for constructing the curves that define the tubes.
Consider the augmented set $\W = [0,1]^n \times [0, t_f]$ and collect $N$ samples $w_r = (\lambda_{1,r}, \ldots, \lambda_{n,r}, \tau_r)$ from $\W$, where $r = [1;N]$. For each $w_r$, define a ball $\W_r $ with radius $\epsilon$, such that for all the points in the augmented set $(\lambda_1, \ldots, \lambda_n, \tau) \in \W$, there exists a $w_r$ that satisfies 
\begin{align}\label{eq:ball}
    \| (\lambda_1, \ldots, \lambda_n, \tau) - w_r \| \leq \epsilon , \forall (\lambda_1, \ldots, \lambda_n, \tau) \in \W_r. 
\end{align}  
This ensures that the union of all these balls forms a superset of the augmented set, i.e., $\bigcup_{r=1}^{N} \W_r \supset [0,1]^n \times [0, t_f]$.

Now, we construct the scenario optimization program (SOP) associated with the ROP:
\begin{subequations} \label{eq:SOP}
\begin{align}
& \min_{[d_1, d_2,...,d_n,\eta]} \quad \eta, \quad  \textrm{s.t.}  \notag \\
& \forall r \in [1;N], (i, \tau_r) \in [1;n] \times [0, t_f]: \notag \\
& \gamma_{i,L}(c_{i,L},\tau_r) - \gamma_{i,U}(c_{i,U},\tau_r) + \gamma_{i,d} \leq \eta, \\
& \dot{\gamma}_{i,L}(c_{i,L},\tau_r) - \mathcal{L}_s \leq \eta, \ 
\dot{\gamma}_{i,U}(c_{i,U},\tau_r) - \mathcal{L}_s \leq \eta;\\
& \forall x :[0,t_f] \rightarrow \R^n, \text{s.t., }\notag\\
&x(\tau_r) = [\lambda_{1,r} \gamma_{1,U}(c_{1,U},\tau_r) + (1 - \lambda_{1,r})\gamma_{1,L}(c_{1,L},\tau_r), \ldots,  \lambda_{n,r} \gamma_{n,U}(c_{n,U},\tau_r) + (1 - \lambda_{n,r})\gamma_{n,L}(c_{n,L},\tau_r)]^\top, \notag \\
&\hspace{8cm} \forall r \in [1;N], \forall (\lambda_{1,r}, \ldots, \lambda_{n,r}, \tau_r) \in [0,1]^n \times [0, t_f], \notag \\
& -\rho^{\phi} (x) \leq \eta \label{subeqn:ROP3}.\\
& d_i = [c_{i,L}, c_{i,U},\eta] \notag.
\end{align}
\end{subequations}

It is clear that the SOP in Equation \eqref{eq:SOP} consists of a finite number of constraints, structurally identical to those in the ROP \eqref{eq:ROP}. However, to ensure that the STTs, constructed by solving the SOP with a finite set of sampled data, satisfy the constraints of the ROP over continuous time and space, we introduce the following assumption:

\begin{assumption} \label{assum:funlip}
    $\gamma_{i,L}(c_{i,L},t)$, $\gamma_{i,U}(c_{i,U},t)$, $\dot{\gamma}_{i,L}(c_{i,L},t)$, $\dot{\gamma}_{i,U}(c_{i,U},t)$
    are all Lipschitz continuous in $t$, with constants $\mathcal{L}_{L}$, $\mathcal{L}_{U}$, $\mathcal{L}_{dL}$ and $\mathcal{L}_{dU}$, respectively, for all $i \in [1;n]$.
\end{assumption}

From \cite{STL_Lip}, we know that there exists $\mathcal{L}_{\rho_\sigma}\in \R^+$, such that for all $x = [x_1, \ldots, x_n]^\top \in X$,
$ \| \rho^{\phi}([x_1, \ldots, x_{k}, \ldots, x_n]^\top) - \rho^{\phi}([x_1, \ldots, \hat{x}_{k}, \ldots, x_n]^\top) \|
     \leq \mathcal{L}_{\rho_\sigma} \| {x}_{k} - \hat{x}_{k} \|$.

\begin{theorem}\label{thm:lip_rob}
    $\mu(\lambda_1, \ldots, \lambda_n)
    := -\rho (x)$ with $x(\tau) = [\lambda_1 \gamma_{1,U}(c_{1,U},\tau) + (1 - \lambda_1)\gamma_{1,L}(c_{1,L},\tau), \ldots,
    \lambda_n \gamma_{n,U}(c_{n,U},\tau) + (1 - \lambda_n)\gamma_{n,L}(c_{n,L},\tau)]^\top$ for all $\tau \in [0,t_f]$, is Lipschitz continuous with the Lipschitz constant $\mathcal{L}_\mu \in \R^+$. 
    That is, there exists $\mathcal{L}_\mu \in \R^+$, such that, for all $ [\lambda_{1}, \ldots, \lambda_{n}]^\top,[\hat{\lambda}_{1}, \ldots, \hat{\lambda}_{n}]^\top \in [0,1]^n$,
    $$
        \|\mu(\lambda_{1}, \ldots, \lambda_{n}) - \mu(\hat{\lambda}_{1}, \ldots, \hat{\lambda}_{n})\| \\
        \leq \mathcal{L}_\mu \| [\lambda_{1}, \ldots, \lambda_{n}]^\top - [\hat{\lambda}_{1}, \ldots, \hat{\lambda}_{n}]^\top \|.    
    $$
\end{theorem}

\begin{proof}
    The proof is provided in the Appendix.
\end{proof}

In Theorem \ref{th:constr}, we show that despite solving the SOP \eqref{eq:SOP} with finitely many sampled points ${(\tau_r, \lambda_r)}_{r=1}^N$, the Lipschitz continuity of the relevant functions allows us to extend the results to the continuous domain $[0, t_f]$ of the ROP \eqref{eq:ROP}.

\begin{theorem} \label{th:constr}
    Suppose that the SOP in (\ref{eq:SOP}) is solved using $N$ sampled data points as defined in Equation \eqref{eq:ball}. Let the optimal value of SOP be $\eta^*$ with the corresponding solution $d_i^* = [c_{i,L}^*, c_{i,U}^*, \eta^*]$, for all $i \in [1;n]$. If 
    \begin{equation} \label{eq:satisfy}
        \eta^* + \mathcal{L}\epsilon \leq 0,
    \end{equation}
    where $\mathcal{L} = \max\{\mathcal{L}_{L} + \mathcal{L}_{U}, \mathcal{L}_{dL}, \mathcal{L}_{dU}, \sqrt{\mathcal{L}_\mu^2+\mathcal{L}_\rho^2 (\mathcal{L}_L+\mathcal{L}_U)^2}$\}, then the STTs $\gamma_{i,L}$ and $\gamma_{i,U}$, for all $i \in [1;n]$, obtained from the SOP in \eqref{eq:SOP} satisfy the conditions of Definition \ref{def:stt}. 
\end{theorem}

\begin{proof}\label{proof:constr}
{
This proof shows that under condition \eqref{eq:satisfy}, the STTs, $\gamma_{i,L}(t)$ and $\gamma_{i,U}(t)$, obtained using SOP in \eqref{eq:SOP} ensures conditions \eqref{subeqn:ROP1}–\eqref{subeqn:ROP3}.
}
Let $\eta^*$ denote the optimal value of the SOP.
    Now from, \eqref{eq:ball}, we infer that $\forall \tau \in [0, t_f], \exists \hspace{0.2em} \tau_r $ such that $|\tau-\tau_r| \leq \epsilon$.
    Thus, $\forall i \in [1;n], \forall r \in [1;N]$, $\forall \tau \in [0, t_f]$:
    \begin{itemize}
        \item[(a)]
            $\gamma_{i,L}(c_{i,L}, \tau) - \gamma_{i,U}(c_{i,U}, \tau) + \gamma_{i,d} =\left(\gamma_{i,L}(c_{i,L}, \tau) - \gamma_{i,L}(c_{i,L},\tau_r)\right) + \big(\gamma_{i,L}(c_{i,L},\tau_r) - 
            \gamma_{i,U}(c_{i,U},\tau_r) + \gamma_{i,d}\big) + \left(\gamma_{i,U}(c_{i,U},\tau_r) - \gamma_{i,U}(c_{i,U}, \tau)\right)
            \leq \mathcal{L}_{L} |t-\tau_r| +\eta^* + \mathcal{L}_{U} |t-\tau_r| \leq \mathcal{L}\epsilon + \eta^* \leq 0,$
        \item[(b)]
            $\dot{\gamma}_{i,L}(c_{i,L}, \tau) - \mathcal{L}_s 
            = \dot{\gamma}_{i,L}(c_{i,L}, \tau) - \dot{\gamma}_{i,L}(c_{i,L}, \tau_r) + \dot{\gamma}_{i,L}(c_{i,L}, \tau_r) - \mathcal{L}_s 
            \leq \mathcal{L}_{dL}\|\tau - \tau_r\| + \eta^* 
            \leq \mathcal{L}\epsilon + \eta^* \leq 0, $
            Similarly, $\dot{\gamma}_{i,U}(c_{i,U}, \tau) - \mathcal{L}_s \leq 0,$
        \item[(c)] Define the trajectories, $x, x_r: [0,t_f] \rightarrow \R^n$ as 
        $$x(\tau) = [\lambda_1 \gamma_{1,U}(c_{1,U},\tau) + (1 - \lambda_1)\gamma_{1,L}(c_{1,L},\tau), \ldots, 
        \lambda_n \gamma_{n,U}(c_{n,U},\tau) + (1 - \lambda_n)\gamma_{n,L}(c_{n,L},\tau)]^\top$$
        $$x_r(\tau_r) = [\lambda_{1,r} \gamma_{1,U}(c_{1,U},\tau_r) + (1 - \lambda_{1,r})\gamma_{1,L}(c_{1,L},\tau_r), \ldots, 
        \lambda_{n,r} \gamma_{n,U}(c_{n,U},\tau_r) + (1 - \lambda_{n,r})\gamma_{n,L}(c_{n,L},\tau_r)]^\top.$$ 
        Let $\tilde{x}_r: [0,t_f] \rightarrow \R^n$ be the continuous reconstruction of $x_r$, such that at any time $\tau$ it holds the value from the nearest sample within $\epsilon$. Then,
        \begin{align*}
         -&\rho^{\phi} (x ) = -\rho^{\phi} (x ) + \rho^{\phi} (\tilde x_r ) - \rho^{\phi} (\tilde x_r ) + \rho^{\phi} (x_r ) -\rho^{\phi} (x_r ) \\
         &= \| \mu(\lambda_1, \ldots, \lambda_n, \tau) - \mu(\lambda_{1,r}, \ldots, \lambda_{n,r}, \tau_r) \| \\ 
         &+ \max_{\tau \in [0,t_f], |\tau-\tau_r|\leq \epsilon} \mathcal{L}_\rho \|\tilde x_r(\tau) - x_r(\tau_r)\| + \eta^* \\
         &\leq \ \eta^* + \mathcal{L}_\mu \|[\lambda_{1}, \ldots, \lambda_{n}]^\top-[\lambda_{1,r}, \ldots, \lambda_{n,r}]^\top\| + \mathcal{L}_\rho(\mathcal{L}_L+\mathcal{L}_U)|\tau-\tau_r| \\
         &\leq \eta^* + \sqrt{\mathcal{L}_\mu^2+\mathcal{L}_\rho^2 (\mathcal{L}_L+\mathcal{L}_U)^2}\epsilon \leq \eta^* + \mathcal{L}\epsilon \leq 0.
    \end{align*}
    \end{itemize}

     Therefore, if the condition in \eqref{eq:satisfy} holds, the STTs defined by the boundary functions $\gamma_{i,L}(c_{i,L}, \tau)$ and $\gamma_{i,U}(c_{i,U}, \tau)$, for all $i \in [1;n]$, as obtained by solving the finitely many constraints in the SOP \eqref{eq:SOP}, satisfies the requirements of Definition \ref{def:stt}, thus completing the proof.
\end{proof}
\begin{remark}
    The Lipschitz constants $\mathcal{L}_L$, $\mathcal{L}_U$, $\mathcal{L}_{dL}$, and $\mathcal{L}_{dU}$ are required to check condition \eqref{eq:satisfy} and Algorithm 1 in the Appendix of \cite{DDSTT_arxiv} gives an estimate of these constants.
\end{remark}

\begin{remark}
    {STTs are shaped by basis functions $p_{i,\con}^k$, such as monomials in $t$ for polynomial-type STT. 
    The computational complexity of SOP depends primarily on the number of unknown coefficients and sampling density. Higher-degree polynomials and finer sampling improve STT flexibility but increase decision variables, leading to polynomial growth in computation time. The change in computational complexity with the number of decision variables and constraints is illustrated in Figure \ref{fig:compute}. However, since the STT construction is offline, the real-time control remains efficient as discussed next.}
\end{remark}
\begin{figure}[h!]
     \centering
     \begin{subfigure}[b]{0.36\textwidth}
         \centering
         \includegraphics[width=\textwidth]{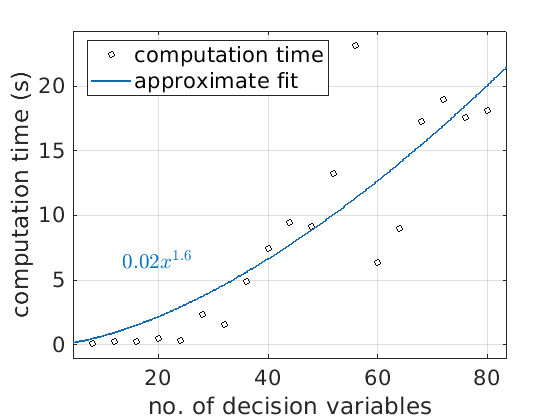}
         \caption{Computation Time vs Decision Variables}
     \end{subfigure}
     \hfill
     \begin{subfigure}[b]{0.36\textwidth}
         \centering
         \includegraphics[width=\textwidth]{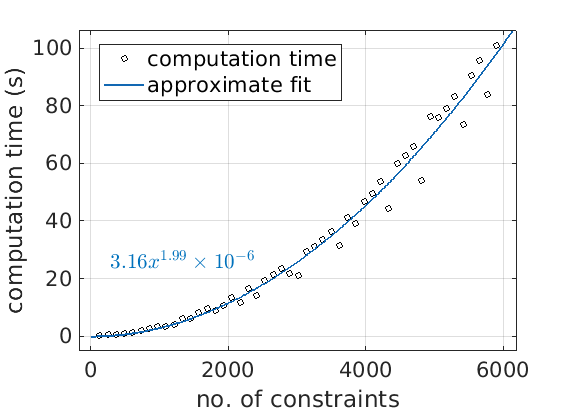}
         \caption{Computation Time vs Constraints}
     \end{subfigure}
        \caption{Time complexity}
        \label{fig:compute}
\end{figure}

\begin{figure}[t]
    \centering
    \includegraphics[width=0.75\textwidth]{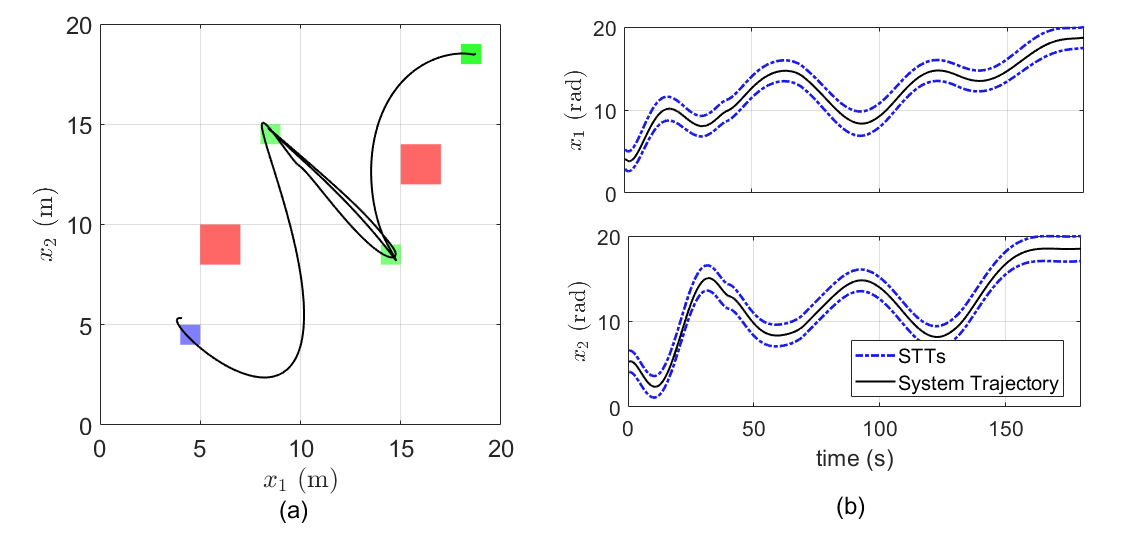}
    \caption{(a) Trajectory (b) Generated STTs for mobile robot.}
    \label{fig:omni}
\end{figure}

\section{Controller Design and Theoretical Analysis}\label{sec:con}
In this section, we derive the control law to satisfy \eqref{eqn:stt_constrain}, thereby satisfying the STL specification. {The controller structure builds upon standard PPC results in \cite{PPC1} and is a direct application of the tube-based control framework in \cite{STT}. We adapt these principles to our setting to solve STL specifications using STT.}
Define the normalized error $e(x,t) := [e_1(x_1,t), \ldots, e_n(x_n,t)]^{\top} = \gamma_m^{-1} (t) \left( 2x - \gamma_s (t) \right)$, the transformed error $\varepsilon(x,t) = \left[\ln\left(\frac{1+e_1(x_1,t)}{1-e_1(x_1,t)}\right), \ldots, \ln\left(\frac{1+e_n(x_n,t)}{1-e_n(x_n,t)}\right) \right]^{\top}$, and the diagonal matrix $\xi(x,t) = 4 \gamma_m^{-1} \left(I_{n}-\text{diag}(e(x,t) \circ e(x,t))\right)^{-1}$,
where, $\gamma_s := [\gamma_{1,U} + \gamma_{1,L}, \ldots, \gamma_{n,U} + \gamma_{n,L}]^{\top}$ and $\gamma_m := \textsf{diag}(\gamma_{1,m},\ldots,\gamma_{n,m})$, with $\gamma_{i,m} = \gamma_{i,U} - \gamma_{i,L}$. 

\begin{theorem} \label{theorem_ras}
    Consider the nonlinear control-affine system $\mathcal{S}$ in \eqref{eqn:sysdyn}. If the initial state $x(0)$ is within the STTs, i.e., $\gamma_{i,L}(0) < x_i(0) < \gamma_{i,U}(0), \forall i \in [1;n]$, then the closed-form control strategy,
    \begin{gather} \label{eqn:Control_stt}
        u(x,t) = -k\xi(x,t) \varepsilon(x,t), k \in \R^+,
    \end{gather}
    will satisfy \eqref{eqn:stt_constrain}, thereby satisfying the STL specification.
\end{theorem}

\begin{proof}
    {Following the proof of Theorem 1 in \cite{STT}}, the control law in Equation \eqref{eqn:Control_stt} ensures that the system trajectory remains within the STTs over $[0, t_f]$. Consequently, from Equations \eqref{eqn:stt_constrain} and \eqref{eqn:stt_stl}, the robustness metric for the STL specification $\phi$ remains positive, guaranteeing $\phi$: $x \models \phi$.
\end{proof}

\begin{remark}
The closed-form time-dependent control law in \eqref{eqn:Control_stt} is approximation-free and guarantees the satisfaction of STL specifications for control-affine systems with unknown dynamics. {The gain $k$ is a user-defined positive design parameter that regulates how aggressively the control law ensures the system state remains within the STTs. Additionally, if $g_{s}(x)$ is negative definite, $k$ (in control law \eqref{eqn:Control_stt}) $\in \R \setminus \R_0^+$.}
\end{remark}

\begin{figure*}
    \centering
    \includegraphics[width=1\textwidth]{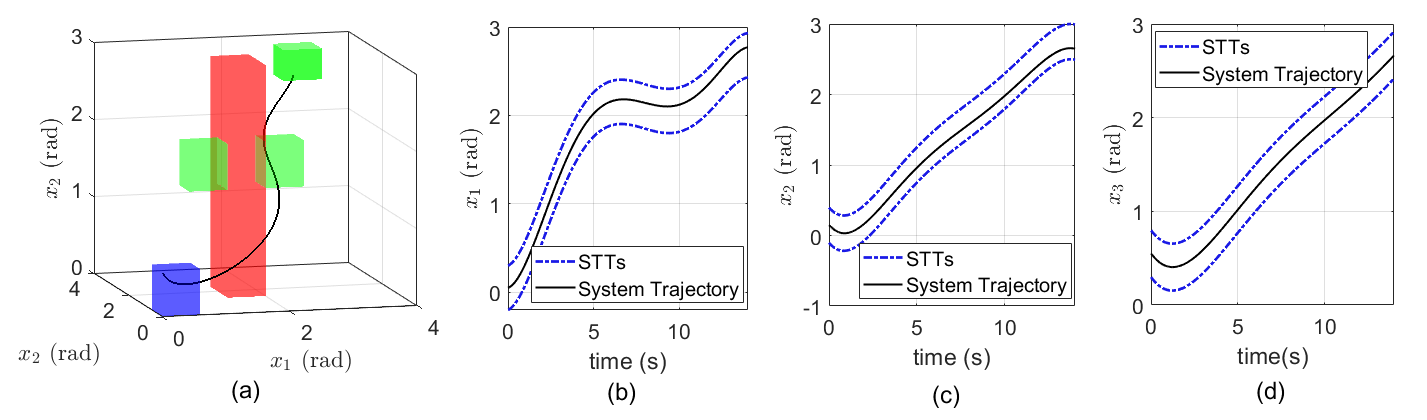}
    \caption{(a) System trajectory and (b), (c), (d) are generated STTs for Spacecraft Case Study.}
    \label{fig:uav}
\end{figure*}

\section{Case Studies}
To demonstrate the efficacy of our STT-based control framework for STL tasks, we used two case studies: (A) an omnidirectional robot and (B) a rotating rigid spacecraft. {For a given candidate value of $\eta$, the feasibility of the SOP can be checked using an SMT solver such as Z3 \cite{z3}. To determine a near-minimum feasible value of $\eta^*$, we use a bisection method over a user-defined interval until convergence within a desired tolerance. The case studies are solved on a Linux Ubuntu operating system with an Intel i7-7700 CPU and 32GB RAM. The implementation code is available at: \href{https://github.com/FocasLab/STL_STT.git}{https://github.com/FocasLab/STL\_STT.git}.}

In case studies, each region $M$ is a hyperrectangle with center $\mathrm{center}(M)$ and size $d_M$, and is represented by the predicate $\tilde{M} := |x - \mathrm{center}(M)|_\infty < d_M$.

\subsection{Omnidirectional Robot}
In this example, we consider the omnidirectional robot model from \cite{Funnel_ICC}
with the following STL task:
\begin{align*}
    \psi = &\G_{[0,130]} ( (\tilde{S} \implies \F_{[30,40]} \tilde{T}_1) \wedge (\tilde{T}_1 \implies \F_{[30,40]} \tilde{T}_2) 
    \wedge (\tilde{T}_2 \implies \F_{[30,40]} \tilde{T}_1)) \\
    &\wedge \F_{[170,180]} \G_{[0,30]} \tilde{G}
    \wedge \G_{[0,200]} \neg (\tilde{O}_1 \vee \tilde{O}_2),    
\end{align*}
\begin{align*}
    \psi = &\G_{[0,13]} ( (\tilde{S} \implies \F_{[3,4]} \tilde{T}_1) \wedge (\tilde{T}_1 \implies \F_{[3,4]} \tilde{T}_2) 
    \wedge (\tilde{T}_2 \implies \F_{[3,4]} \tilde{T}_1)) \wedge \F_{[17,18]} \G_{[0,3]} \tilde{G}
    \wedge \G_{[0,20]} \neg (\tilde{O}_1 \vee \tilde{O}_2),    
\end{align*}
with $S = [4,5]\times[4,5]$, $T_1 = [8,9]\times[14,15]$, $T_2 = [14,15]\times[8,9]$, $G = [18,19]\times[18,19]$, $O_1 = [5,7]\times[8,10]$, and $O_2 = [15,17]\times[12,14]$.
This STL task requires the robot to switch between $T_1$ and $T_2$ every 3–4 time units until time 13, reach and stay in region $G$ from time 17 to 20, and avoid $O_1$ and $O_2$ for the first 20 time units, starting from $S$. {Piecewise polynomial tubes of degree 5 were generated in about 1 min, with control synthesis requiring 0.02 seconds. The Lipschitz constant was $\mathcal{L} = 0.88$, with a grid size of $\epsilon = 0.1$. The computed $\eta^* = -0.12$ satisfies \eqref{eq:satisfy} $\eta^* + \mathcal{L}\epsilon = -0.032 \leq 0$.} The STTs and the robot trajectory are presented in Figure \ref{fig:omni}.

\subsection{Rotating Rigid Spacecraft}
In this case study, we validate the proposed method on a rigid spacecraft, adopted from \cite{Khalil}, with the following specification for its dynamics:
$$\psi = \tilde{S} \implies \F_{[7, 8]} (\tilde{T}_1 \vee \tilde{T}_2) \wedge \F_{[14,15]} \tilde{G} \wedge \G_{[0,15]} \neg \tilde{O},$$
with $S = [0,0.6] \times [0,0.6] \times [0.4,1]$, $T_1 = [2,2.6] \times [1.4,2] \times [1.4,2]$, $T_2 = [0.8,1.4] \times [1.4,2] \times [1.4,2]$, $G = [2.6,3.2] \times [2.6,3.2] \times [2.6,3.2]$, and $O = [1.4,2] \times [1.4,2.4] \times [0,3]$.
This defines a sequential STL task, in which the spacecraft, starting from $S$, must eventually reach either region $T_1$ or $T_2$ between 7 and 8 seconds. Then, it must "eventually" proceed to the goal region $G$ between 14 and 15 seconds. Additionally, the spacecraft must avoid the unsafe set $O$ at all times from 0 to 15 seconds. {Polynomial tubes of degree 5 were generated in 0.91 seconds with controller synthesis requiring about 0.03 seconds. The Lipschitz constant was $\mathcal{L} = 0.6$, with a grid size of $\epsilon = 0.5$. The computed $\eta^* = -0.31$ satisfies \eqref{eq:satisfy} $\eta^* + \mathcal{L}\epsilon = -0.01 \leq 0$.} The STTs and the system trajectory are shown in Figure \ref{fig:uav}.

\begin{table}[h!]
\centering
\caption{Control synthesis time (seconds)}\label{table:computation_time}
\begin{tabular}{|l|l|l|l|l|l|}
\hline
Case    & \textbf{STT} & MILP & MPC  & CBF & Funnel \\ \hline
stlcg-1 & \textbf{0.019}   & 271.773  & 372.313  & N/A     & N/A        \\ \hline
stlcg-2 & \textbf{0.022}   & 23.150   & 29.862   & 18.418  & 0.451      \\ \hline
space-1 & \textbf{0.038}   & 7938.420 & 3179.483 & N/A     & N/A        \\ \hline
space-2 & \textbf{0.033}   & 200.510  & 290.192  & 80.308  & 0.453      \\ \hline
\end{tabular}
\end{table}

\subsection{Comparison}

The comparison is carried out on two case studies: \textit{stlcg-1}, a two-dimensional task presented in \cite{STL_grad} and benchmarked in \cite{STL_multi_chuchu}, and \textit{space-1}, the three-dimensional task in Case Study B. Both involve complex STL specifications, while \textit{stlcg-2} and \textit{space-2} are their simplified counterparts with only STL fragments \cite{CBF_STL_Dimos, STL_PPC}. From Table \ref{table:computation_time}, STT achieves significantly lower control synthesis times than MILP, MPC, and CBF, all of which are computed at an interval of 0.01s, scale poorly as the state dimension increases and the STL specifications become more complex. In contrast, STT maintains minimal and nearly constant computational complexity, demonstrating its scalability. Although funnel-based methods perform competitively, both CBF and funnel fail to handle tasks involving negation or disjunction (\textit{stlcg-1} and \textit{space-1}). 

\section{Conclusion and Future Work}\label{sec:conclusion}
This work presented an STT-based control framework for satisfying STL specifications in control-affine systems with unknown dynamics. By formulating STL constraints as a robust optimization problem and solving it through scenario optimization, we constructed STTs with formal correctness guarantees. The approximation-free closed-form control law constrains the system trajectory within the STTs, ensuring STL satisfaction without requiring real-time optimization, significantly improving computational efficiency and scalability compared to existing methods. Compared with existing methods, STTs require less time to synthesize controllers and generate efficient trajectories for complex STL specifications.

{While the proposed approach effectively handles general STL tasks, the control effort remains unbounded. Future work will focus on incorporating prescribed input bounds. We also plan to improve the computational efficiency of the STT construction process and extend the framework to underactuated systems.}

\bibliographystyle{unsrt} 
\bibliography{sources} 

\subsection*{Appendix A. Proof of Theorem \ref{thm:lip_rob}}
\begin{proof}[of Theorem \ref{thm:lip_rob}]
We show that $\mu(\theta, \lambda)
    := -\rho^{\phi} \big( x \big)$ with $x(\tau) = [\lambda_1 \gamma_{1,U}(c_{1,U},\tau) + (1 - \lambda_1)\gamma_{1,L}(c_{1,L},\tau), \ldots,
    \lambda_n \gamma_{n,U}(c_{n,U},\tau) + (1 - \lambda_n)\gamma_{n,L}(c_{n,L},\tau)]^\top$ for all $\tau \in [0,t_f]$, is Lipschitz continuous with respect to $\lambda$.    
    
    From \cite{STL_Lip}, we know that there exists $\mathcal{L}_{\rho}\in \R^+$, such that
    $$| \rho^{\phi}(x_1, {.\kern 0em.\kern 0em.}, x_{k}, {.\kern 0em.\kern 0em.}, x_n) - \rho^{\phi}(x_1, {.\kern 0em.\kern 0em.}, \hat{x}_{k}, {.\kern 0em.\kern 0em.}, x_n) | \leq \mathcal{L}_{\rho} | {x}_{k} - \hat{x}_{k} |.$$

    Further, note that $\overline{\gamma}_{L}:=\max_{t\in[0,t_f]}\max_{i\in[1;n]}\gamma_{L}(c_{i,L},t)$ and $\overline{\gamma}_{U}:=\max_{t\in[0,t_f]}\max_{i\in[1;n]}\gamma_{U}(c_{i,U},t)$ are finite because $\gamma_{i,L}(\cdot)$ and $\gamma_{i,U}(\cdot)$ are Lipschitz on the compact interval $[0,t_f]$. 
     
    \begin{itemize}
        \item[Step 1:] 
        For some $\lambda_i \in [0,1]$, $i\in [1;n]$, fix rest $\lambda_j \in [0,1]$, for all $j\in [1;n]\setminus\{i\}$: 
        \begin{align*}
            & \| \mu(\lambda_1, \ldots, \lambda_i, \ldots, \lambda_n) - \mu(\lambda_1, \ldots, \hat\lambda_i, \ldots, \lambda_n) \| 
            \leq \mathcal{L}_\rho (\overline{\gamma}_{L}+\overline{\gamma}_{U}) |\lambda_i - \hat\lambda_i|
        \end{align*}
        
        \item[Step 2] 
        \begin{align*}
            &\|\mu(\lambda_{1}, \ldots, \lambda_{n}) - \mu(\hat{\lambda}_{1}, \ldots, \hat{\lambda}_{n})\| \\
            = & \|\mu(\lambda_{1}, \lambda_{2}, \ldots, \lambda_{n}) 
            - \mu(\hat\lambda_{1}, \lambda_{2}, \ldots, \lambda_{n})
            + \mu(\hat\lambda_{1}, \lambda_{2}, \ldots, \lambda_{n})
            - \mu(\hat\lambda_{1}, \hat\lambda_{2}, \ldots, \lambda_{n})\\
            &\qquad + \mu(\hat\lambda_{1}, \hat\lambda_{2}, \ldots, \lambda_{n})
            \ldots
            - \mu(\hat{\lambda}_{1}, \hat\lambda_2 \ldots, \hat{\lambda}_{n})\| \\
            \leq & \mathcal{L}_\rho (\overline{\gamma}_{L}+\overline{\gamma}_{U}) |\lambda_1 - \hat\lambda_1| + \mathcal{L}_\rho (\overline{\gamma}_{L}+\overline{\gamma}_{U}) |\lambda_2 - \hat\lambda_2| + \ldots + \mathcal{L}_\rho (\overline{\gamma}_{L}+\overline{\gamma}_{U}) |\lambda_n - \hat\lambda_n| \\
            \leq & \mathcal{L}_\rho (\overline{\gamma}_{L}+\overline{\gamma}_{U}) \sum_{i=1}^{n} |\lambda_i - \hat\lambda_i|   
        \end{align*}
    \end{itemize}
    From Cauchy-Schwarz inequality,
    \begin{align}
        \leq \mathcal{L}_\mu \| [\lambda_{1}, \ldots, \lambda_{n}]^\top - [\hat{\lambda}_{1}, \ldots, \hat{\lambda}_{n}]^\top \|.  
    \end{align}
    where 
    $\mathcal{L}_\mu = \mathcal{L}_\rho (\overline{\gamma}_{L}+\overline{\gamma}_{U}) \sqrt{n}$.
\end{proof}

\end{document}